\documentclass{article}

\usepackage{arxiv}

\usepackage[utf8]{inputenc} 
\usepackage[T1]{fontenc}    
\usepackage[english]{babel}
\usepackage[colorlinks=true,allcolors=blue]{hyperref} 
\usepackage{url}            
\usepackage{booktabs}       
\usepackage{array}
\usepackage{blkarray}
\usepackage{amsmath,amsfonts,amssymb,amsthm} 
\newtheorem{theorem}{Theorem}
\newtheorem{proposition}[theorem]{Proposition}
\usepackage{nicefrac}       
\usepackage{microtype}      
\usepackage{graphicx}
\usepackage[superscript,biblabel,nomove]{cite}
\usepackage{jabbrv}
\usepackage{doi}

\usepackage{mathptmx}   
\usepackage{helvet}
\usepackage{courier}
\usepackage{newtx}

\makeatletter
  \g@addto@macro\@uclclist{%
    \alpha\Alpha%
    \beta\Beta%
    \gamma\Gamma%
    \delta\Delta%
    \epsilon\Epsilon%
    \zeta\Zeta%
    \eta\Eta%
    \theta\Theta%
    \iota\Iota%
    \kappa\Kappa%
    \lambda\Lambda%
    \nu\Nu%
    \xi\Xi%
    \pi\Pi%
    \rho\Rho%
    \sigma\Sigma%
    \tau\Tau%
    \upsilon\Upsilon%
    \phi\Phi%
    \chi\Chi%
    \psi\Psi%
    \omega\Omega%
  }
  \newcommand\Alpha{\mathrm{A}}
  \newcommand\Beta{\mathrm{B}}
  \newcommand\Epsilon{\mathrm{E}}
  \newcommand\Zeta{\mathrm{Z}}
  \newcommand\Eta{\mathrm{H}}
  \newcommand\Iota{\mathrm{I}}
  \newcommand\Kappa{\mathrm{K}}
  \newcommand\Nu{\mathrm{N}}
  \newcommand\Rho{\mathrm{P}}
  \newcommand\Tau{\mathrm{T}}
  \newcommand\Chi{\mathrm{X}}
\makeatother

\newcommand{\ecipci}{ECI-PCI algorithm}
\newcommand{\mean}[1]{\langle #1 \rangle} 
\newcommand{\sd}[1]{\mathrm{sd}(#1)} 
\newcommand{\mtx}[1]{\boldsymbol{\mathbf{\MakeUppercase{#1}}}} 
\newcommand{\vect}[1]{\boldsymbol{\mathbf{\MakeLowercase{#1}}}} 
\newcommand{\grp}[1]{\mathcal{\MakeUppercase{#1}}} 
\newcommand{\elij}[2]{\MakeUppercase{#1}_{#2}} 
\newcommand{\eli}[2]{\MakeLowercase{#1}_{#2}} 
\newcommand{\nth}[2]{\vect{#1}_{#2}} 
\newcommand{\prob}{\mathbb{P}} 
\newcommand{\lpls}{l} 
\newcommand{\specs}{m} 
\newcommand{\spec}{\mtx{\specs}} 
\newcommand{\divtys}{d} 
\newcommand{\ubitys}{u} 
\newcommand{\divty}{\mtx{\divtys}} 
\newcommand{\ubity}{\mtx{\ubitys}} 
\newcommand{\clyr}[1]{\widetilde{#1}} 
\newcommand{\plyr}[1]{\widehat{#1}} 
\newcommand{\cprojs}{\clyr{\specs}} 
\newcommand{\pprojs}{\plyr{\specs}} 
\newcommand{\cproj}{\mtx{\cprojs}} 
\newcommand{\pproj}{\mtx{\pprojs}} 
\newcommand{\sims}{s} 
\newcommand{\csims}{\clyr{\sims}} 
\newcommand{\psims}{\plyr{\sims}} 
\newcommand{\csim}{\mtx{\csims}} 
\newcommand{\psim}{\mtx{\psims}} 
\newcommand{\rw}[1]{\mtx{#1}_\mathrm{rw}} 
\newcommand{\sym}[1]{\mtx{#1}_\mathrm{sym}} 
\newcommand{\clpl}{\clyr{\lpls}} 
\newcommand{\ceigs}{d} 
\newcommand{\peigs}{u} 
\newcommand{\ceig}{\vect{\ceigs}} 
\newcommand{\peig}{\vect{\peigs}} 
\newcommand{\eigvals}{\lambda} 
\newcommand{\csvs}{x} 
\newcommand{\psvs}{y} 
\newcommand{\svals}{\sigma} 
\newcommand{\mcsv}{\mtx{\csvs}} 
\newcommand{\mpsv}{\mtx{\psvs}} 
\newcommand{\msval}{\mtx{\svals}} 
\newcommand{\csv}{\vect{\csvs}} 
\newcommand{\psv}{\vect{\psvs}} 
\newcommand{\clusts}{z} 
\newcommand{\cclusts}{\clyr{\clusts}} 
\newcommand{\pclusts}{\plyr{\clusts}} 
\newcommand{\clust}{\vect{\clusts}} 
\newcommand{\cclust}{\vect{\cclusts}} 
\newcommand{\pclust}{\vect{\pclusts}} 
\newcommand{\ecis}{\mathrm{ECI}} 
\newcommand{\pcis}{\mathrm{PCI}} 

\DeclareRobustCommand{\VAN}[2]{#2}

\title{Reinterpreting Economic Complexity: A co-clustering approach}

\author{
  \href{https://orcid.org/0000-0002-2878-9360}{\includegraphics[scale=0.06]{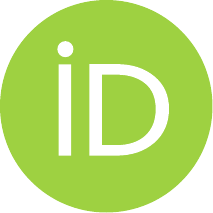}\hspace{1mm}Carlo Bottai} \\
        Department of Economics, Management and Statistics\\
        University of Milano–Bicocca\\
        Milano, Italy\\
        \texttt{carlo.bottai@unimib.it} \\
  \And
  \href{https://orcid.org/0000-0001-9268-446X}{\includegraphics[scale=0.06]{orcid.pdf}\hspace{1mm}Jacopo Di~Iorio} \\
	Department of Statistics\\
        Penn State University\\
        University Park (PA), USA\\
        \texttt{jqd5830@psu.edu} \\
  \AND
  \href{https://orcid.org/0000-0001-6825-2165}{\includegraphics[scale=0.06]{orcid.pdf}\hspace{1mm}Martina Iori} \\
	Institute of Economics \& L'EMbeDS\\
        Sant'Anna School of Advanced Studies\\
        Pisa, Italy\\
        \texttt{martina.iori@santannapisa.it} \\
}

\renewcommand{\shorttitle}{Reinterpreting Economic Complexity}

\hypersetup{
  pdftitle={Reinterpreting Economic Complexity: A co-clustering approach},
  pdfsubject={econ.GN, stat.AP},
  pdfauthor={Carlo Bottai, Jacopo Di Iorio, Martina Iori},
  pdfkeywords={Economic Development, Economic Completely, Capabilities, Spectral Co-clustering},
}

\begin{document}
\maketitle

\begin{abstract}
  Economic growth results from countries' accumulation of organizational and technological capabilities. The Economic and Product Complexity Indices, introduced as an attempt to measure these capabilities from a country's basket of exported products, have become popular to study economic development, the geography of innovation, and industrial policies. Despite this reception, the interpretation of these indicators proved difficult. Although the original Method of Reflections suggested a direct interconnection between country and product metrics, it has been proved that the Economic and Product Complexity Indices result from a spectral clustering algorithm that separately groups similar countries or similar products, respectively. This recent approach to economic and product complexity conflicts with the original one and treats separately countries and products. However, building on previous interpretations of the indices and the recent evolution in spectral clustering, we show that these indices simultaneously identify two co-clusters of similar countries and products. This viewpoint reconciles the spectral clustering interpretation of the indices with the original Method of Reflections interpretation. By proving the often neglected intimate relationship between country and product complexity, this approach emphasizes the role of a selected set of products in determining economic development while extending the range of applications of these indicators in economics.
\end{abstract}


\section*{Introduction}
In the process of economic development, countries accumulate organizational and technological capabilities that allow them to diversify into and competitively export new goods and services of increasing technological sophistication.\cite{FagerbergEtAlInnovation2010, HausmannEconomic2016, FrenkenEtAlCapabilities2023} Therefore, export baskets including technologically advanced products usually reflect the accumulation of rare and specialized capabilities by the exporters. Based on this intuition, a series of indicators have been proposed to infer a country's development stage from trade data.\cite{LallTechnological1992, HausmannEtAlWhat2007} In 2009, Hidalgo and Hausmann\cite{HidalgoHausmannBuilding2009, HausmannEtAlAtlas2014} introduced a now widespread indicator that examines the structural properties of bipartite graphs representing the monetary values of products annually exported by countries as a proxy for the unobserved capabilities of those countries: i.e., a tripartite network where the middle layer is unobservable (\emph{latent}); see Figures~\ref{fig:tri_bi_graphs}a and \ref{fig:tri_bi_graphs}b. This methodology extracts information from the exports' bipartite graph by summarizing the productive activities carried out in each country into the Economic Complexity Index (ECI) and the capabilities necessary to produce and competitively export each product into the Product Complexity Index (PCI). While the two indicators are firmly bounded and simultaneously defined, the literature mostly neglected the PCI, with notable exceptions.\cite{FelipeEtAlProduct2012, MealyTeytelboymEconomic2022} On the contrary, scholars interested in countries' development found in the ECI a powerful tool for capturing and forecasting living standard differences across countries and regions, as measured by their GDP per capita.\cite{HidalgoEconomic2021, BallandEtAlNew2022}

\begin{figure}[bp]
  \centering
  \includegraphics[width=0.9\textwidth]{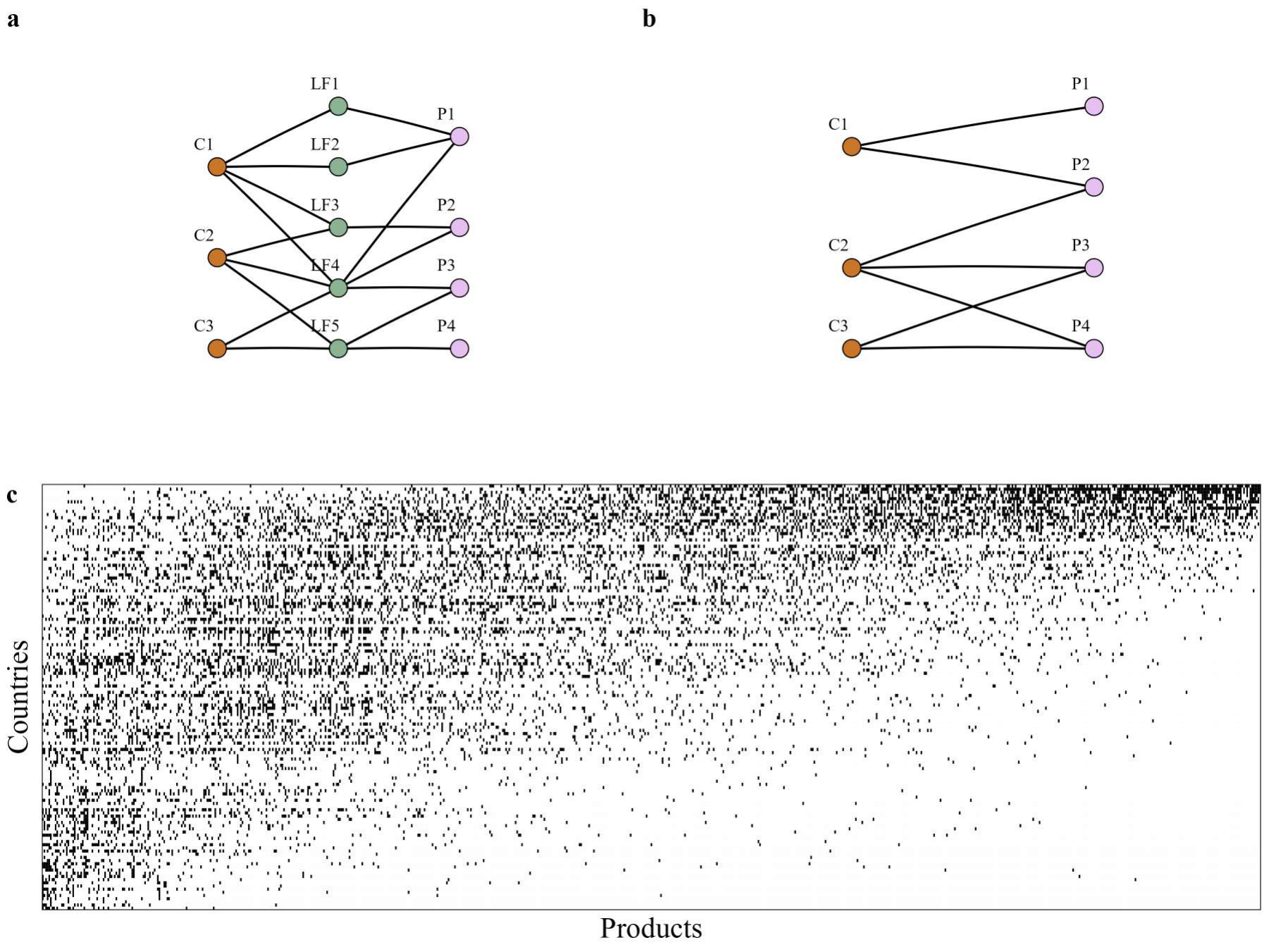}
  \caption{International trade export graphs and specialization matrix. International trade export graphs show the connections between countries and the products they competitively produce. Panel (a) reports an example of a tripartite trade graph with a hidden layer between countries ($C$) and products ($P$); i.e. a latent factor like countries' capabilities. The hidden layer identifies the latent factors $LF$ owned by a specific country $C$ that are necessary to competitively export a certain product $P$. This tripartite graph originates the observed bipartite graph of international trade exports reported in panel (b). This latter includes only country ($C$) and product ($P$) layers and directly shows the link between countries and exported products. In this example, $C1$ is the only country competitively exporting product $P1$ because it owns the latent factors $LF1$ and $LF2$ necessary to produce $P1$, as reported in panel (a). Panel (c) shows the bi-adjacency matrix derived from the international trade export graph, also known as the specialization matrix ($\mathbf{M}$), obtained using international trade data in 1990 (more details in Methods). This matrix assumes value 1 (represented by a dark square) when a given country (row) exports competitively a certain product (column) and 0 otherwise.}
  \label{fig:tri_bi_graphs}
\end{figure}

As ECI's applications spread, a debate on its interpretation and definition has emerged. On the one hand, a stream of literature has focused on the refinement of the index or its extension to other settings.\cite{TacchellaEtAlNew2012, CristelliEtAlMeasuring2013, AlbeaikEtAlImproving2017, GneccoEtAlMachine2022, LiuGaoNormality2022, InouaSimple2023, SciarraEtAlReconciling2020} On the other, increasing efforts have been devoted to interpreting the algorithm resulting in the ECI and PCI.\cite{CaldarelliEtAlNetwork2012, InouaSimple2023, Gomez-LievanoMethods2018, MealyEtAlInterpreting2019, vanDamEtAlCorrespondence2021, HidalgoEconomic2021, TezaEtAlEntropic2021, SchetterMeasure2022, McNerneyEtAlBridging2023, ServedioEtAlEconomic2024} 
The original interpretation, also known as Method of Reflections (MoR),\cite{HidalgoHausmannBuilding2009} conceives the ECI as a measure of \emph{generalized capability diversity}. In this framework, the more a country has diversified into technologically advanced products, whose production requires a large variety of rare capabilities, the more complex (i.e., with higher ECI) it is considered. These technologically advanced products are those considered more complex (i.e., with higher PCI) being produced solely by countries endowed with capabilities necessary for advanced technologies and organizational routines.
However, recent contributions reject ECI as a measure of capability diversity.\cite{InouaSimple2023} On the contrary, it has been shown that ECI (PCI) is equivalent to a dimensionality reduction algorithm that groups countries (products) with similar export baskets (exporter sets). Mealy and coauthors\cite{MealyEtAlInterpreting2019} proved that ECI and PCI result from the application to the bipartite export graph of various statistical methodologies, such as \emph{spectral clustering},\cite{ShiMalikNormalized2000, DingTutorial2004, vonLuxburgTutorial2007} \emph{correspondence analysis},\cite{BenzecriAnalyse1973, HillReciprocal1973, HillCorrespondence1974, GreenacreTheory1984, GreenacreCorrespondence2017} and \emph{diffusion maps}.\cite{NadlerEtAlDiffusion2005}
Considering the nature of these techniques, the relative position in the ECI (PCI) ranking of any two countries (products) can be seen as measuring their similarity rather than their relative complexity. 
In line with this view, van Dam and coauthors,\cite{vanDamEtAlCorrespondence2021} emphasized that ECI and GDP per capita correlate because ``countries with similar export baskets also have similar wealth''.
Similarly, Gomez-Lievano\cite{Gomez-LievanoMethods2018} represents the two indices under the lens of a \emph{recommender system} that embeds the principles of ``collective imitation'' between countries (ECI) and ``relatedness'' between products (PCI). According to these principles, a country $c$ is more likely to start producing a product $p$ the more $c$ is similar to the country $c'$ already producing $p$. At the same time, $c$ is more likely to introduce $p$ when $p$ is close to the product $p'$ already produced by $c$. Over time, the resulting consensus dynamics processes align export baskets and exporter sets across countries and products.
Building on the principle of relatedness too, McNerney and coauthors\cite{McNerneyEtAlBridging2023} interpreted ECI as measuring structural change, i.e. describing baskets of economic activities in a low-dimension space emerging from economic diversification dynamics driven by countries' current specialization baskets.
Thus, we can interpret the ECI (PCI) algorithm as a procedure that returns the best one-dimensional indicator, ranking countries (products) based on their relative similarity. The PCI comes first, determining the coordinates of the long-term directions of economic change, while ECI follows, showing where each country is located along this evolutionary path.

Undoubtedly, emerging interpretations of ECI as a measure of similarity between export baskets conflict with the original MoR interpretation. By relying on the idea of generalized capability diversity, the MoR qualifies products that must be competitively exported by a country to climb along the economic development stages, highlighting the close relationship between ECI and PCI. On the contrary, interpreting ECI (PCI) as similarity is vague about which set of products (countries) identifies two countries (products) as similar. In our view, this is due to the omission of the network latent layer, corresponding to countries' capabilities, in the similarity-measure interpretation, as already pointed out---moving from a different perspective---by van Dam and coauthors.\cite{vanDamEtAlCorrespondence2021}

In this paper, we reconcile these conflicting interpretations of ECI and PCI. Specifically, we preserve the interpretation of the algorithm as a measure of country similarity while providing an integrated view of the two indices that explicitly connects countries and products, as in the original MoR interpretation. To do so, we prove that the \ecipci{} is equivalent to a spectral co-clustering algorithm\cite{DhillonCoClustering2001} that gathers, at the same time, products and countries together. This procedure returns two mutually exclusive co-clusters employing a concept of duality between the rows and columns of the graph's bi-adjacency matrix: row clusters depend on their column distributions, and column clusters are induced by co-occurrence in rows. Since ECI and PCI extract information from a bipartite graph of international exports, co-clustering applies to this setting better than clustering, which considers each set of nodes (countries or products) separately. In this context, a co-clustering approach allows for the simultaneous partition of countries and products into two disjoint sets of similar and interrelated product-country pairs. As shown in the following sections, these disjoint sets can be interpreted as more technologically advanced product-country pairs and less advanced ones. While the similarity-measure interpretation applies as well in the co-clustering context, our approach helps clarify which products (countries) influence the clustering process among countries (products) more. As discussed later, the proximity between any two graph nodes reflects a random walker's two-step transitions on the exports' bipartite graph that, in the clustering context, are hidden in the matrix algebra.\cite{YenEtAlLink2011, MealyEtAlInterpreting2019, vanDamEtAlCorrespondence2021} Even further, in the co-clustering framework, we can hypothesize latent factors (e.g., capabilities, as in the MoR) inducing the two clustering processes simultaneously.\cite{GreenacreTheory1984, KlugerSpectral2003, FoussEtAlBipartite2016, vanDamEtAlCorrespondence2021} 

To further clarify the role of products in economic development, we also propose a generalization of the co-clustering algorithm by using a soft (or fuzzy) clustering technique that returns posterior probabilities to belong to the assigned co-cluster. This procedure allows identifying countries and products with a higher likelihood of moving upward or downward in the ECI and PCI ranking and, therefore, more interesting in terms of economic and production dynamics. It also helps to understand which products better characterize the groups of more or less advanced countries in different periods.

The proposed approach to the definition of ECI and PCI allows a better understanding of the profound interconnection between economic development and production systems. To provide an example of the possible implications, we apply our methodology to global trade data, i.e., the original Hidalgo and Hausmann's setting.

\section*{Results}

\subsection*{Deriving ECI and PCI from a co-clustering algorithm}
The \ecipci{} works on a binary contingency table $\spec_{m \times n} = [\elij{\specs}{cp}]$ that represents a bipartite graph $\grp{\specs}$ measuring whether a country $c$ produces competitively a product $p$ ($\elij{\specs}{cp} = 1$) or not ($\elij{\specs}{cp} = 0$), as Figure~\ref{fig:tri_bi_graphs} illustrates. The number of products exported competitively by a country (\emph{diversity}) is stored in the diagonal matrix $\divty = \mathrm{diag}(\sum_p \elij{\specs}{cp})$, while the number of countries that competitively produce a product (\emph{ubiquity}) is collected in the diagonal matrix $\ubity = \mathrm{diag}(\sum_c \elij{\specs}{cp})$. This algorithm is equivalent to solving the eigenvalue problem of two random-walk normalized similarity matrices, $\cproj = \divty^{-1} \, \csim = \divty^{-1} \, (\spec \, \ubity^{-1} \, \spec^{\transp})$ and $\pproj = \ubity^{-1} \, \psim = \ubity^{-1} \, (\spec^{\transp} \, \divty^{-1} \, \spec)$, and leading to ECI and PCI, respectively.\cite{CaldarelliEtAlNetwork2012, HausmannEtAlAtlas2014}

Mealy and coauthors\cite{MealyEtAlInterpreting2019} explained that the ECI-PCI algorithm is equivalent to a spectral clustering technique partitioning two similarity graphs---$\csim$ and $\psim$---each into a pair of disjoint sets---$\mathcal{A}$ and $\mathcal{B}$---that are internally similar and externally dissimilar. Following this interpretation, the algorithm divides countries (products) into two groups leveraging the similarity between countries (products) in terms of their export specialization (concentration) patterns.

The spectral clustering interpretation deals with ECI and PCI separately, without directly embedding their mathematical and economic relationship. In what follows, we propose an interpretation of the ECI-PCI algorithm based on \emph{spectral co-clustering} that makes explicit the interwoven connection between ECI and PCI by computing them simultaneously from the same singular value decomposition (SVD) procedure.

In the clustering context, each element $\elij{\cprojs}{cc'}$ of $\cproj$ represents the probability $\prob_{c \rightarrow c'}$ of a random walker to move from $c$ to $c'$ in one step, given the ubiquity of the products competitively produced by $c$ or $c'$ and each element $\elij{\pprojs}{pp'}$ of $\pproj$ the corresponding probability $\prob_{p \rightarrow p'}$.\cite{HidalgoHausmannBuilding2009, CaldarelliEtAlNetwork2012, MealyEtAlInterpreting2019} As explained in the Methods in full detail, through stochastic complementation,\cite{MeyerStochastic1989, YenEtAlLink2011} it is possible to show that these two transition probabilities represent the reduced Markov chains, $\clyr{\grp{\specs}}$ and $\plyr{\grp{\specs}}$, of the chain $\grp{\specs}$ associated with the transition matrix
\[
\rw{\xi} = 
\begin{bmatrix}
  \mtx{0} & \cproj \\
  \pproj^{\transp} & \mtx{0}
\end{bmatrix},
\]
so that the values in $\cproj$ can be read as one-step transition probabilities of $\clyr{\grp{\specs}}_{c} \rightarrow \clyr{\grp{\specs}}_{c'}$ on $\clyr{\grp{\specs}}$ or as two-step transition probabilities of $\grp{\specs}_{c} \Rightarrow \grp{\specs}_{c'}$ on $\grp{\specs}$ (i.e., on the bipartite graph in Figure~\ref{fig:tri_bi_graphs}b); and the same applies to $\pproj$.
In this framework, the ECI-PCI algorithm can be seen as the solution of the eigenvalue problem $\rw{\xi} \, \clust = (1 - \eigvals) \, \clust$ defined on the bipartite graph $\grp{\specs}$, instead of the solution of two distinct problems defined on its monopartite projections $\clyr{\grp{\specs}}$ and $\plyr{\grp{\specs}}$. Consequently, the two indices derive from a single matrix $\rw{\xi}$ and not from two distinct matrices, as suggested by previous interpretations of the algorithm.

This result highlights a duality between countries and products that can be exploited and preserved by applying a spectral co-clustering algorithm instead of a spectral clustering one. In fact, in the spectral co-clustering algorithm proposed by Dhillon,\cite{DhillonCoClustering2001} row clustering induces column clustering, and \emph{vice versa}. The algorithm can be applied to any matrix $\spec$ and returns two mutually exclusive co-clusters. The bi-partition problem is solved using the second left- and right-singular vectors $\nth{\csv}{2}$ and $\nth{\psv}{2}$ resulting from the SVD of $\sym{\specs} = \divty^{-1/2} \, \spec \, \ubity^{-1/2}$; i.e., the symmetric normalization of $\spec$. The SVD of the real matrix is a factorization of the form $\sym{\specs} = \mcsv \, \msval \, \mpsv^{\transp}$, where $\mcsv$ and $\mpsv$ are orthogonal unitary matrices, and $\msval$ is a $ m \times n$ rectangular diagonal matrix whose entries $\eli{\svals}{i} = \elij{\svals}{ii}$ are non-negative real numbers named singular values. The columns of $\mcsv$ and $\mpsv$ are called \emph{left-} and \emph{right-singular vectors} of $\sym{\specs}$ and the corresponding sets of columns $[\nth{\csv}{1}, \nth{\csv}{2}, \dots, \nth{\csv}{m}]$ and $[\nth{\psv}{1}, \nth{\psv}{2}, \dots, \nth{\psv}{n}]$ are two sets of orthonormal bases. A clustering technique, such as $k$-means, is then applied to $\nth{\clusts}{2} = [\nth{\cclusts}{2}, \nth{\pclusts}{2}] = [\divty^{-1/2} \, \nth{\csv}{2}, \ubity^{-1/2} \, \nth{\psv}{2}]$ to identify the two clusters. Considering that the clustering is performed simultaneously on the reduced representation of rows and columns, the two obtained clusters are co-clusters formed by sets of both rows and columns. It is possible to prove that the results of the spectral co-clustering algorithm correspond to bipartite simultaneously the non-standardized ECI ($\nth{\ceigs}{2} = \nth{\cclusts}{2}$) and PCI ($\nth{\peigs}{2} = \nth{\svals}{2}^{-1} \, \nth{\pclusts}{2}$) in two disjoint sets of countries and products. See Supplementary Information (SI) for the full theoretical proof.

While equivalent to the clustering interpretation of the algorithm, the co-clustering approach offers additional insights into the mechanisms that group together countries and products. If we consider the bi-incidence matrix of the graph $\grp{\specs}$, $[\mtx{R}, \mtx{C}]^{\transp}$, we can define the specialization matrix $\spec$ as $\spec = \mtx{R}^{\transp} \, \mtx{C}$, where each element $R_{ki}$ and $C_{kj}$ of the incidence matrices $\mtx{R}$ and $\mtx{C}$ of the bi-graph $\grp{\spec}$ assume value one if country $c_i$ competitively produces product $p_j$, and zero otherwise. Then, we can assume a latent factor explaining the specializations observed in $\mtx{R}$ and concentrations observed in $\mtx{C}$ to read the ECI and PCI as measuring this latent factor.\cite{GreenacreTheory1984, FoussEtAlBipartite2016, vanDamEtAlCorrespondence2021} Let us call this latent factor ``capabilities'' and suppose they are expressed as an intensity. Let us also assume that countries with a given capabilities intensity tend to specialize in products whose production requires a similar capabilities intensity and that products requiring a certain capabilities intensity concentrate into countries able to mobilize such capabilities level. In this context, the \ecipci{} aims to identify the latent scores that better explain the observed $\spec$. Moving from a vertices- to an edges-centered view, the algorithm maximizes the correlation between the latent scores for the countries and products, as shown by Figure~\ref{fig:country_product_probs_cca}e, so that in this latent space, countries endowed with some capabilities intensity will be close to products requiring a similar capabilities intensity to be produced competitively. See Methods for further details.

\begin{figure}[bp]
  \centering
  \includegraphics[width=0.9\textwidth]{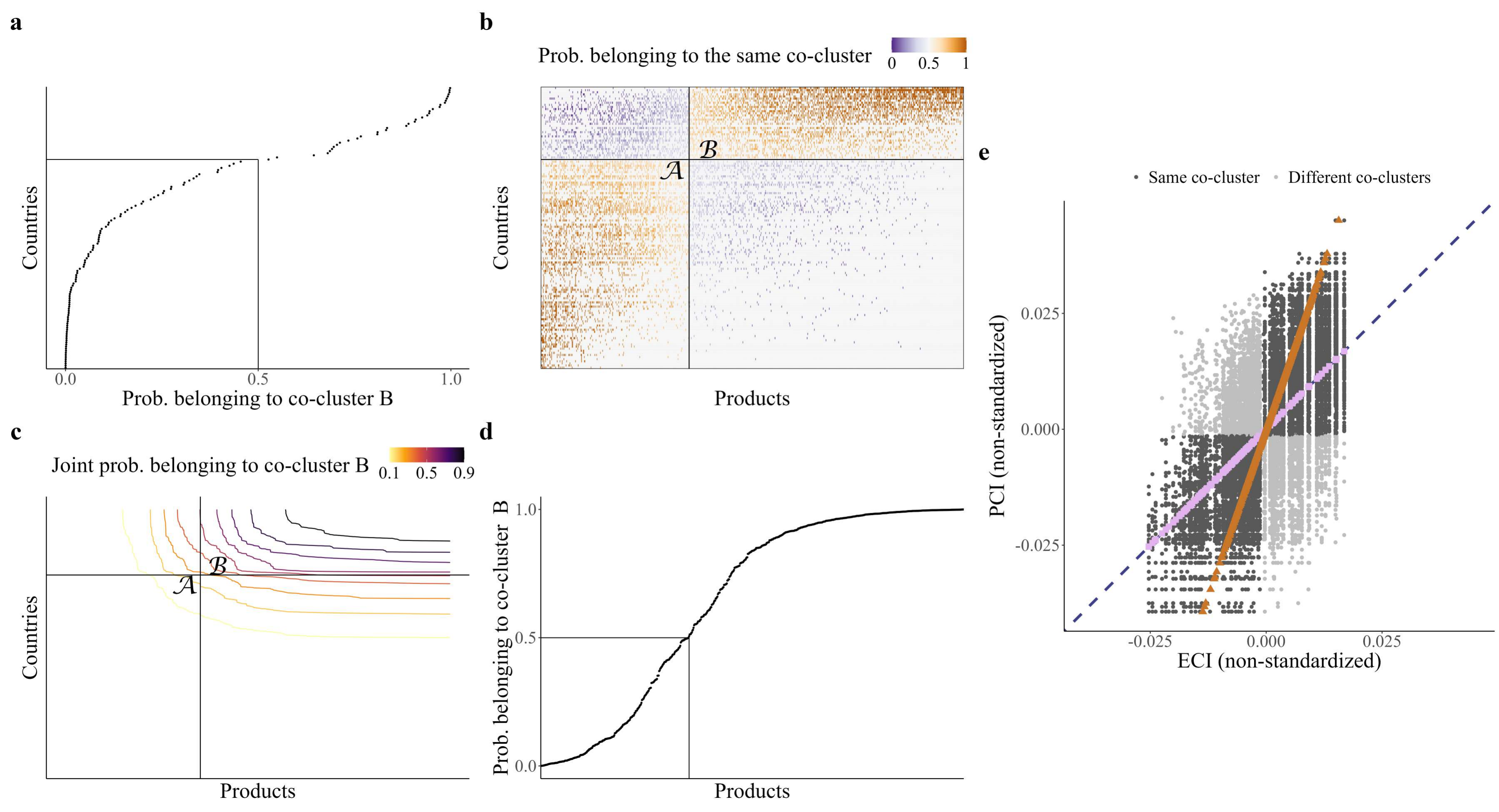}
  \caption{Probabilities of belonging to co-clusters and non-standardized ECI and PCI of each country-product pair in 1990. The fuzzy co-clustering algorithm assigns each country and each product to co-clusters $\mathcal{A}$ and $\mathcal{B}$ with a certain probability. When the probability of belonging to a specific co-cluster is higher than 0.5, countries and products are assigned to that co-cluster. Panel (a) shows the distribution of the probability of belonging to co-cluster $\mathcal{B}$ for countries. Countries with probability higher than 0.5 are assigned to cluster $\mathcal{B}$, and all other countries are assigned to cluster $\mathcal{A}$. The same occurs for products, as shown in panel (d). The two co-clusters, composed of both countries and products, are highlighted in orange at the top right and bottom left corners of the specialization matrix $\mathbf{M}$ displayed in panel (b). Pairs that do not belong to the same co-cluster are purple. The intensity of colors reflects the joint probability of belonging to the same co-cluster for each country-product pair. Panel (c) reports the joint probabilities of belonging to co-cluster $\mathcal{B}$ for countries (rows) and products (columns). As in panel (b), co-clusters $\mathcal{A}$ and $\mathcal{B}$ are marked in the bottom left and top right corners, respectively. For the purpose of this graph, the joint probability of belonging to co-cluster $\mathcal{B}$ has been computed under the hypothesis of independence. Panel (e) reports the non-standardized ECI and PCI of each country-product pair. The darker dots represent pairs of countries and products that belong to the same co-cluster. In blue is the diagonal of the Euclidean space. The orange triangles indicate the average PCI of products ($x$-axis) competitively produced by countries with a given ECI ($y$-axis). The purple squares represent the average ECI of the countries where each product over-concentrates. In all figures, we plot the non-standardized ECI. However, readers should be aware that the authors of the original ECI-PCI algorithm defined $\ecis \coloneq \big( \nth{\ceigs}{2} - \mean{\nth{\ceigs}{2}} \big) \big/ \sd{\nth{\ceigs}{2}}$ and $\pcis \coloneq \big( \nth{\peigs}{2} - \mean{\nth{\ceigs}{2}} \big) \big/ \sd{\nth{\ceigs}{2}}$.\cite{HidalgoHausmannBuilding2009, CIDEcomplexity2014} See SI for a discussion about the implications of this standardization procedure in a co-clustering context.}
  \label{fig:country_product_probs_cca}
\end{figure}

We can conclude that the \ecipci{} co-clusters countries and products in two balanced and disjoint sets that are internally similar and externally dissimilar. Then, a country $c$, picked at random within the ones that competitively produce the product $p$, and $p$ itself will belong to the same co-cluster with a high chance. Alternatively, we can conclude that the co-clusters result from the presence of latent features (e.g., capabilities) explaining both countries' specialization and products' concentration patterns.

The duality between ECI and PCI has several consequences. First, it is not possible to properly study economic development via ECI without considering the features of countries' production (PCI). Second, similar countries are grouped together because they export products with specific characteristics, i.e., products in their same co-cluster. From this perspective, PCI identifies products with greater (or lower) value for economic development, since they are in the cluster of more (or less) technologically advanced countries. At the same time, ECI identifies countries that competitively export more (or less) technologically advanced products.

To better understand this duality and extend the readability of the two indices, we modify the spectral-clustering algorithm by using a soft clustering technique after the usual SVD. We employ a clustering based on a Gaussian mixture model to provide posterior probabilities that summarize how well countries or products fit their co-cluster; see Methods for further details.\cite{ScruccaEtAlModel2023} As a result, together with the ECI and PCI scores, we also obtain the cluster to which countries and products belong associated with the probability of being in that cluster.

Besides making explicit the interrelation between ECI and PCI, the proposed co-clustering algorithm offers a more efficient way to obtain the indices. While computing and eigen-decomposing $\csim$ requires a $\mathcal{O}(nm^2)$ and a $\mathcal{O}(m^3)$ operations just to obtain the ECI, SVD-decomposing $\sym{\specs}$ requires a $\mathcal{O}(mn^2)$ operation, resulting in both indices simultaneously.

\subsection*{The interrelation between country and product complexity indices}
In this section, we employ our method to analyze the international trade data (see Methods). Figure \ref{fig:country_product_probs_cca}b shows the two resulting co-clusters, namely $\mathcal{A}$ and $\mathcal{B}$, embedded in the original contingency table $\spec$ for the year 1990. In this figure, the rows and columns of $\spec$ have been arranged based on their ECI and PCI, respectively, and colors indicate the probability of belonging to the same co-cluster. Co-cluster $\mathcal{A}$, in the bottom-left corner, comprises both low-ECI countries and low-PCI products; conversely, co-cluster $\mathcal{B}$ encompasses the remaining more complex countries and products. Colors highlight how some countries or products are well fitting their co-cluster while others---those at the boundary between $\mathcal{A}$ and $\mathcal{B}$---are not. The marginal probability of belonging to co-cluster $\mathcal{B}$ is shown in Figures \ref{fig:country_product_probs_cca}a and \ref{fig:country_product_probs_cca}d for countries and products, respectively, showing that only items with a probability higher than 0.5 are assigned to this co-cluster.

The interrelation between ECI and PCI, which is reflected in the belonging of countries and products to the same co-cluster, is also highlighted in Figure \ref{fig:country_product_probs_cca}e. Each dot represents a country-product pair. When the country and the product belong to the same co-cluster (either $\mathcal{A}$ or $\mathcal{B}$), the pair dot is colored in dark gray, while light gray indicates a disagreement between the two co-cluster classifications. As in the previous figure, co-clusters $\mathcal{A}$ and $\mathcal{B}$ are in the bottom-left and top-right quadrants of the plot, respectively. The figure highlights a positive correlation of 0.59 between ECI and PCI, as products with high PCI are competitively produced by countries with a high ECI on average, as shown by the orange triangles representing the average ECI of countries exporting products with a given PCI. Similarly, countries with high ECI specialize in products with, on average, high PCI, as emphasized by the purple squares overlapping the diagonal of the plane. This figure confirms that, in the ECI-PCI framework, countries and products are not disjoint elements but are indissolubly bounded in the inner features of the algorithm. It also suggests the presence of a hidden layer able to simultaneously account for countries' specialization and products' concentration patterns, as explained in the Methods.

This relationship is even more relevant when we move to a dynamic analysis. If we compare the co-clusters' composition in different years, we observe changes in countries based on their export baskets; see the SI for a full overview. In 1970, the composition of co-cluster $\mathcal{A}$ was defined by African and South Asian countries that specialized in inedible crude materials or animal and vegetable oils, fats and waxes. In the new century, this co-cluster have been instead dominated by producers of crude petroleum, oil and mineral fuels. Co-cluster $\mathcal{B}$ is more stable over time for what concerns countries---all exporting specialized machines---, with a stable presence of Germany and Switzerland across the entire period, with Japan joining from the 1980s. However, high-PCI products have changed over time. This category was initially characterized by tools with electric motors, including dish washing machines and transport equipment. More recently, it comprises optical instruments, photography, and tools from the metal and chemistry industries. In the new century, the composition of top ECI countries has also changed for the first time, with European countries overtaken by emerging top-Asian countries, including Taiwan and South Korea.

To further stress the strong interrelation between ECI and PCI, we propose two simulation exercises to study how ECI and PCI change when countries specialize in products associated with different PCI levels. We based our simulations on international trade data for 1990.
In the first simulation, countries virtually specialize in products they were not exporting before. For the sake of simplicity, we select only three sets of products: products belonging to co-cluster $\mathcal{A}$ with high probability (higher than 0.997), those firmly belonging to co-cluster $\mathcal{B}$ (probability higher than 0.997), and borderline products (i.e., those with a probability of belonging to their co-cluster lower than 0.6). To analyze the results of these simulation exercises, we first focus on how a product's PCI varies when an additional country starts competitively producing it. As shown in Figure \ref{fig:sim_change}a, this change is inversely proportional to the product's starting PCI. Low-PCI products tend to increase their PCI since they are, on average, added to countries with an ECI higher than countries usually exporting them (low-ECI countries). Borderline products slightly decrease their PCI, on average, once they have been added to new countries because they have a high chance of landing in low-ECI countries (co-cluster $\mathcal{A}$ is much larger than co-cluster $\mathcal{B}$; see Figure \ref{fig:country_product_probs_cca}). However, the more pronounced decrease in PCI is associated with high-PCI products. The PCI of these products is determined by being exported exclusively by high-ECI countries, as it suddenly drops as soon as countries with lower ECI specialize in it. To complete this picture, we test how new specializations are reflected in the ECI of countries adopting these new products. Figure \ref{fig:sim_change}b shows that competitively exporting new products increases countries' ECI, on average, except when the products' PCI is particularly low (i.e., for products with a high probability of belonging to co-cluster $\mathcal{A}$). As expected, ECI gain grows with the added product's PCI, but, on average, countries also benefit from adding borderline products. Nevertheless, the effect of new specialization is heterogeneous across countries, depending on their initial development level, as shown in Figure \ref{fig:sim_change}c. Low-ECI countries in co-cluster $\mathcal{A}$ benefit the most, on average, from new specializations, and their gain is maximum when they specialize in high-PCI products from co-cluster $\mathcal{B}$. Those countries also show the highest variability in ECI change, with a considerable decrease in ECI when they add another low-PCI product to their export basket. The advantages for high-ECI countries are marginal as they experience limited ECI variations, but they must avoid specializing in low-PCI products to preserve their status. The interconnection between ECI and PCI is even more evident when we analyze simulation results for single countries. Figure \ref{fig:sim_change}d shows how Japan's ECI, the highest in our sample for 1990, varies when we introduce new specializations. Specifically, it increases only with the introduction of products firmly belonging to co-cluster $\mathcal{B}$ and decreases otherwise. Interestingly, borderline products, independently of their initial co-cluster, move to the co-cluster $\mathcal{B}$ once they have been added to Japan's export basket. When we replicate this analysis for a borderline country, such as Brazil (Figure \ref{fig:sim_change}e), its ECI increases for all co-cluster $\mathcal{B}$ and borderline products. This new specialization has almost no effect on the new products' cluster belonging. Finally, the introduction of new products in the export basket of Guinea---the country with the lowest ECI for 1990---usually results in a substantial increase in the ECI, except when it adds a few products firmly belonging to co-cluster $\mathcal{A}$. As for Japan, the new specialization of Guinea affects the new product's PCI, moving all borderline products to the co-cluster $\mathcal{A}$.

\begin{figure}[tp]
  \centering
  \includegraphics[width=0.85\textwidth]{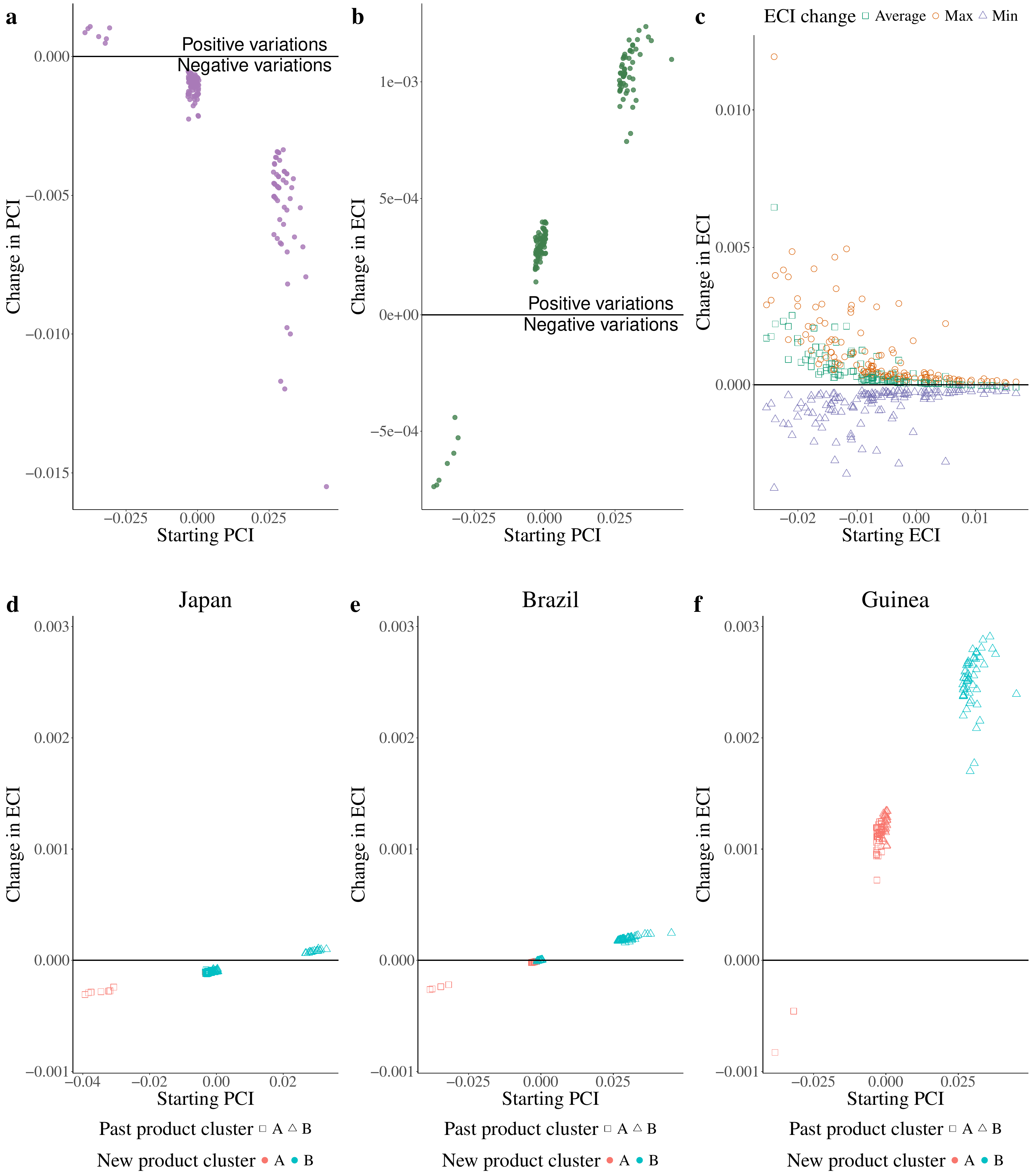}
  \caption{Changes in non-standardized ECI and PCI when countries specialize in new products in 1990. These plots summarize data resulting from multiple simulations in which we add an additional specialization to a country's export basket. Note that we add specializations depending on countries' previous export baskets. Therefore, the list of potential specializations to add varies across countries. Panel (a) reports the average PCI change of products newly introduced by counties with respect to their starting PCI value. The average ECI change of countries specializing in these new products, depending on their starting PCI, is shown in panel (b). Panel (c) summarizes the range of possible changes in ECI---its maximum, average, and minimum---following the specialization of a country in a new product with respect to the starting country's ECI. Panels (d)-(f) report how three countries (Japan, Brazil, and Guinea) modify their ECI depending on the ECI of the product they introduce. These plots also show how products' PCI move from co-cluster $\mathcal{A}$ to co-cluster $\mathcal{B}$ and \textit{vice versa} depending on the country starting to export them.}
  \label{fig:sim_change}
\end{figure}

The second simulation iteratively modifies the export basket of Guinea (identified by the ISO3 code GIN). In this simulation, we iteratively add the product that maximizes the country's ECI until it is impossible to find other products increasing the ECI. As reported in Figure \ref{fig:sim_gin}, Guinea's ECI experiences rapid growth by adding products that are among those with the highest PCI, such as chemical products, tools and machines for specialized industries or welding, brazing, cutting, and working metals. This country's ECI growth results in an increase in the PCI of products it was exporting before and, in turn, in the ECI of other countries specialized in those products, such as Nigeria (NGA), Angola (AGO), and Sierra Leone (SLE). The effect is more pronounced for countries whose export basket overlaps the most with Guinea's one. While some countries see their ECI ranking growing, none reaches high ECI values since their export basket does not vary during the simulation. At the same time, Guinea's export basket resembles more and more high-ECI countries' export baskets, moving away from low-ranking countries' specializations. ECI of countries in co-cluster $\mathcal{B}$ is almost unaffected.

\begin{figure}[tbp]
  \centering
  \includegraphics[width=0.85\textwidth]{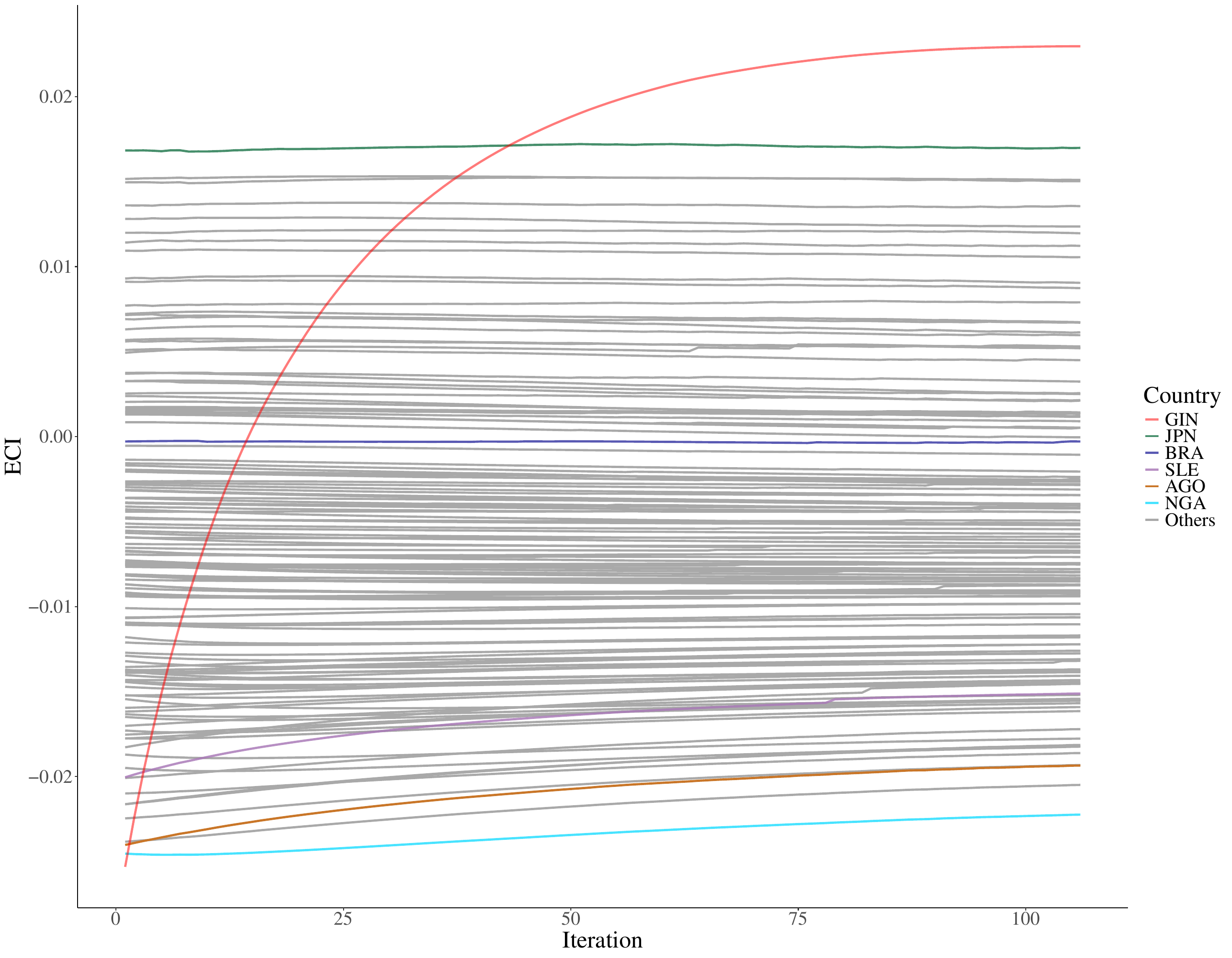}
  \caption{ECI ranking resulting from adding high-PCI products to Guinea's export basket. In this simulation, we iteratively specialize Guinea (GIN, in red) in the product that maximizes its ECI. The simulation stops after 106 iterations when no products would increase Guinea's ECI. Other countries' export baskets are unchanged throughout the simulation. In this simulation, Guinea overtakes all other countries after 44 iterations and reaches its maximum at the 106th iteration. The added products are among those with the highest PCI, especially at the beginning of the simulation, with a very high probability (above 0.9865 in general and above 0.9924 for the first 10 iterations) to belong to co-cluster $\mathcal{B}$. }
  \label{fig:sim_gin}
\end{figure}

\section*{Discussion}
Measuring and understanding national and regional economic development is crucial for economic scholars and policymakers. As evidenced by their widespread diffusion, the ECI and other ``complexity metrics'' have proved to be solid allies in achieving this purpose since they provide a simple means to capture complex economic dynamics from easily accessible data. However, the uncertainty on the interpretation of these indicators and discussions on methodological aspects have undermined their reputation and pose challenges to their further diffusion in economic analyses. 

In this paper, we proved that the \ecipci{} is equivalent to a co-clustering algorithm simultaneously grouping countries and products into two sets, based on how similar they are to one another concerning some hidden unobservable features. This co-clustering approach offers several advantages---including efficiency from an algorithm viewpoint---and opens new avenues for the understanding of economic development and its long-term drivers. Most importantly, it reconciles the two mostly known interpretations of ECI. Even though aligned with the similarity-measure interpretation, the co-clustering approach suggests the presence of a hidden layer in the international exports network that, in line with the original MoR interpretation, could be seen as countries' capabilities. Differently from recent literature advancements,\cite{McNerneyEtAlBridging2023} the co-clustering perspective suggests a co-determination of ECI and PCI, that so result intrinsically connected to one another. Therefore, ECI is more than an index summarizing complex information from trade data, but only the ECI-PCI combination helps to discern between products more or less supportive of long-term economic development. A high-PCI product is one that, given the current overall network structure, is competitively exported only by more economically and technologically developed countries. This suggests that the capabilities needed for its production are relatively rare and distant from more widespread ones. These capabilities may change over time or be absorbed by other countries too, resulting in variations in PCI scores. Consequently, the analysis of a country's economic development cannot be disjoint by the study of what they are capable of producing. This suggests the need for a more prominent role of PCI in contributions exploiting ECI as a measure of economic development. All in all, the integrated view proposed in this paper suggests an evolutionary perspective on complexity metrics and brings back production to the core of the debate on economic growth. This aspect is particularly relevant in current times, characterized by crucial economic challenges, such as the green transition, that can lead to profound labor-market transformation and severely affect countries' competitiveness. In this context, identifying products and capabilities prone to economic development may help mitigate the adverse effects of these impelling transformations.

\section*{Methods}

\subsection*{International trade data} \label{sec:data}
As aforementioned, the empirical parts of the article use international trade data from UN~COMTRADE as cleaned by The Growth Lab at Harvard University; see \url{https://atlas.cid.harvard.edu/} for further details.\cite{UNComtrade, GrowthLabInternational2019} To maximize the temporal consistency of the data, we decided to eliminate countries and products not constantly observed, as well as any trade of services. Specifically, we analyze the value, in current USD, of 763 goods exported by 134 different countries between 1962 and 2021.

The products are classified according to the SITC Rev.~2 classification (4-digit detail level). On the one hand, this classification allows us to exploit a very long time series, covering sixty years of observations. On the other, to preserve temporal consistency, SITC generates issues with new goods segments.\cite{GrowthLabInternational2019} These choices do not affect the analysis exposed in this article, since our results depend on statistical properties of the indicators, which are independent of data characteristics. 

\subsection*{\ecipci{} in a clustering context}
The ECI-PCI framework aims to measure the capabilities of national economies using information on their export baskets (see Figure \ref{fig:tri_bi_graphs}). The algorithm works on a binary contingency table $\spec = [\elij{\specs}{cp}]$ that represents a bipartite graph $\grp{\specs}$ measuring whether a country $c$ produces competitively a product $p$ ($\elij{\specs}{cp} = 1$) or not ($\elij{\specs}{cp} = 0$).  The number of products produced competitively by a country, said \emph{diversity}, is stored in the diagonal matrix $\divty = \mathrm{diag}(\sum_p \elij{\specs}{cp})$. Instead, the number of countries that competitively produce a product, referred to as \emph{ubiquity}, is collected in the diagonal matrix $\ubity = \mathrm{diag}(\sum_c \elij{\specs}{cp})$. The MoR algorithm recursively uses \emph{diversity} and \emph{ubiquity} to correct each other. After some iterations, the algorithm converges to what is known as ECI and PCI.\cite{HidalgoHausmannBuilding2009} This algorithm is equivalent to solving the eigenvalue problem of two random-walk normalized similarity matrices: $\cproj = \divty^{-1} \, \csim = \divty^{-1} \, (\spec \, \ubity^{-1} \, \spec^{\transp})$ and $\pproj = \divty^{-1} \, \psim = \ubity^{-1} \, (\spec^{\transp} \, \divty^{-1} \, \spec)$.\cite{CaldarelliEtAlNetwork2012, HausmannEtAlAtlas2014} In this framework, $\elij{\cprojs}{cc'}$ is the probability $\prob_{c \rightarrow c'}$ of a random walker to move from $c$ to $c'$ in one step, given the ubiquity of the products competitively produced by $c$ or $c'$, while $\elij{\pprojs}{pp'}$ is the corresponding probability $\prob_{p \rightarrow p'}$.\cite{CaldarelliEtAlNetwork2012, MealyEtAlInterpreting2019}

The contribution by Mealy and coauthors\cite{MealyEtAlInterpreting2019} showed that the ECI-PCI algorithm is equivalent to a spectral clustering technique. Limiting the analysis to bipartitions, spectral clustering partitions a similarity graph into two disjoint sets---$\mathcal{A}$ and $\mathcal{B}$, with $\mathcal{A} \cap \mathcal{B} = \emptyset$---that are internally similar and externally dissimilar. Specifically, the normalized-cut criterion (Ncut) minimizes $\mathrm{Ncut}(\mathcal{A}, \mathcal{B}) =  \prob_{\mathcal{A} \, \rightarrow \, \mathcal{B}} + \prob_{\mathcal{B} \, \rightarrow \, \mathcal{A}}$, where $\prob_{\mathcal{A} \rightarrow \mathcal{B}} = \sum_{c \in \mathcal{A}, \, c' \in \mathcal{B}} \elij{\csims}{cc'} \big/ \sum_{c \in \mathcal{A}, \, p} \elij{\specs}{cp}$ is the probability of a random walker to transit from partition $\mathcal{A}$ to $\mathcal{B}$ in one step, while exploring the graphs $\clyr{\grp{\specs}}$ or $\plyr{\grp{\specs}}$.\cite{ShiMalikNormalized2000, MeilaShiRandom2001, vonLuxburgTutorial2007} This clustering technique proceeds in three steps. First, it transforms the data contained in $\spec$ into a pairwise similarity matrix $\csim = \spec \, \ubity^{-1} \, \spec^{\transp}$. Second, it eigendecomposes the normalized Laplacian matrix of $\csim$, which is $\rw{\clpl} = \divty^{-1} \, (\divty - \csim)$. Third, it uses standard clustering techniques, such as $k$-means, on the second smallest eigenvector of $\rw{\clpl}$, $\nth{\ceigs}{2}$, to cluster the data. Given that the (non-standardized) ECI is, by definition, the second eigenvector of $\cproj = \divty^{-1} \, \csim$, it is equivalent to the approximate solution of the spectral problem formulated by the Ncut algorithm, $\rw{\clpl} \, \ceig = \eigvals \, \ceig$, simply ordering the eigenvalues from the largest to the smallest instead of the opposite, since $\cproj \, \ceig = (1 - \eigvals) \, \ceig$. Symmetrical considerations apply to the PCI. Following this interpretation, the algorithm divides countries (products) into two groups based on how similar two countries (products) are to each other in terms of their export specialization (concentration) pattern.

\subsection*{\ecipci{} in a co-clustering context}
However, through stochastic complementation,\cite{MeyerStochastic1989, YenEtAlLink2011} it is possible to show that each random walker's \emph{one-step} transition probability $\prob_{c \rightarrow c'}$ ($\prob_{p \rightarrow p'}$) associated with the transition matrix $\cproj$ ($\pproj$)---defined on the monopartite projection $\clyr{\grp{\specs}}$ ($\plyr{\grp{\specs}}$) of the graph $\grp{\specs}$ on the countries' (products') layer---actually reflects a \emph{two-step} transition probability $\prob_{c \Rightarrow c'}$ ($\prob_{p \Rightarrow p'}$) on the transition matrix $\rw{\xi}$---defined on the bipartite graph $\grp{\specs}$---, where $\rw{\xi}$ is defined as
\[
\rw{\xi} = \mtx{\Delta}^{-1} \, \mtx{\xi} = 
\begin{bmatrix}
  \divty^{-1} & \mtx{0} \\
  \mtx{0} & \ubity^{-1}
\end{bmatrix}
\begin{bmatrix}
  \csim & \mtx{0} \\
  \mtx{0} & \psim^{\transp}
\end{bmatrix}
=
\begin{bmatrix}
  \divty^{-1} & \mtx{0} \\
  \mtx{0} & \ubity^{-1}
\end{bmatrix}
\left(
\begin{bmatrix}
  \mtx{0} & \spec \\
  \spec^{\transp} & \mtx{0}
\end{bmatrix}
\begin{bmatrix}
  \divty^{-1} & \mtx{0} \\
  \mtx{0} & \ubity^{-1}
\end{bmatrix}
\begin{bmatrix}
  \mtx{0} & \spec \\
  \spec^{\transp} & \mtx{0}
\end{bmatrix}
\right)
=
\begin{bmatrix}
  \mtx{0} & \cproj \\
  \pproj^{\transp} & \mtx{0}
\end{bmatrix},
\]

Following this interpretation, the algorithm divides countries and products into two groups based on how likely it is for a random walker to move from partition $\mathcal{A}$ to $\mathcal{B}$ in two steps, while exploring the exports' bipartite graph $\grp{\specs}$. That is, two countries will belong to the same group depending on how similar they are to each other in terms of their export specialization pattern, how similar the products they competitively produce are to each other in terms of their exporters' concentration patterns, and so on along spiral-like reasoning.

This is equivalent to co-clustering the specialization matrix, simultaneously accounting for both the specialization and concentration patterns in the data. Co-clustering (or bi-clustering) is a set of unsupervised learning techniques that perform simultaneous clustering on both dimensions of a data matrix to identify co-clusters, i.e., subsets of rows that exhibit similar behaviors across a subset of columns and \emph{vice versa}. While Hartigan\cite{HartiganDirect1972} introduced the idea of co-clustering in 1972, its first implementation occurred in 2000 by Cheng and Church\cite{ChengChurchBiclustering2000} to discover co-clusters in gene array expression data. Since then, many other algorithms have been created by embedding different strategies and algorithmic concepts.\cite{MadeiraOliveiraBiclustering2004, PontesEtAlBiclustering2015} Among others, Dhillon,\cite{DhillonCoClustering2001} Zha et al.,\cite{ZhaEtAtBipartite2001} and Kluger et al.\cite{KlugerSpectral2003} identify co-clusters using bipartite spectral graph partitioning, relying on Singular Value Decomposition (SVD).

First applied on a bipartite graph between $m$ documents and $n$ words represented by an $n \times m$ adjacency matrix $\spec$, the procedure introduced by Dhillon returns two mutually exclusive co-clusters based on a concept of duality between the rows and columns of the bi-adjacency matrix of the graph: i.e., row clusters depend on their column distributions, and column clusters are determined by co-occurrence in rows. This is achieved by applying the normalized Ncut criterion and identifying the node bi-partition $(\mathcal{A}, \mathcal{B})$ that minimizes the crossing edges weight between the two partitions. This NP-hard problem can be approximately solved by bi-partitioning the vector $\nth{\clusts}{2} = [\nth{\cclusts}{2}, \nth{\pclusts}{2}] = [\divty^{-1/2}\nth{\csv}{2}, \ubity^{-1/2}\nth{\psv}{2}]$ where $\nth{\csv}{2}$ and $\nth{\psv}{2}$ are the second left and right singular vectors obtained from the SVD of $\spec_{sym} = \divty^{-1/2} \spec \ubity^{-1/2}$. The bipartitioning is then performed using the $k$-means algorithm on $\nth{\clusts}{2}$. Considering that the clustering is performed simultaneously on the reduced representation of rows and columns, the two obtained clusters are co-clusters formed by sets of both rows (e.g., documents or countries) and columns (words or products).

In this paper, instead of applying $k$-means on $\nth{\clusts}{2}$, we perform the bi-partitioning using a Gaussian mixture model (GMM).\cite{ScruccaEtAlModel2023} GMM is a probabilistic model based on the assumption that the observations are generated from a mixture of a finite number of Gaussian distributions (in our case, two since we are interested in a bi-partition) whose parameters are unknown a priori. These models can be thought of as a generalization of $k$-means and incorporate information about the covariance structure of the data as well as the centers of the latent Gaussians. The fitting of the mixture of Gaussian models is performed by expectation maximization (EM). Differently from $k$-means---which is a hard clustering method, associating each point to one and only one cluster---, Gaussian mixture models return for each observation the probability of being associated with any cluster. Therefore, our approach returns a probability of belonging to co-clusters $\mathcal{A}$ and $\mathcal{B}$ for each row and column of the matrix.

Starting from this spectral co-clustering strategy, Kluger and coauthors \cite{KlugerSpectral2003} introduced an algorithm for clustering tumor profiles collected via RNA microarrays. Translated into the terminology of our paper, Kluger and coauthors assume that the bi-adjacency matrix $M$ is the result of a noisy process depending on three factors: a hidden factor $\elij{E}{cp}$, the \emph{diversity} $\elij{\divtys}{cc}$ of country $c$, and a latent factor $\eli{\peigs}{p}$ representing the overall tendency of countries to specialize in product $p$. This setting corresponds to partitioning rows and columns under the assumption that the data has an underlying checkerboard structure which is represented by the hidden factor $\elij{E}{cp}$. From an algorithmic point of view, the algorithms proposed by Dhillon (and consequently us) and Kluger and coauthors differ only for the type of matrix normalization applied to make the checkerboard structure of $\spec$ more evident.

\subsection*{\ecipci{} latent factor interpretation}
Based on the results by Kluger and coauthors, an alternative, more economically sound interpretation of the co-clustering framework proposed, may be considered. Assuming that countries at the same economic development stage over- or under-specialize in a subset of products, we can apply a co-clustering strategy to detect these specialization patterns in the data. Indeed, as Kluger and coauthors explain,\cite{KlugerSpectral2003} we can assume the matrix $\spec$ being a noisy representation of a checkerboard-structured matrix $\mtx{E} = [\elij{E}{cp}]$ whose entries are determined by a latent factor driving the specialization phenomenon just mentioned. In an idealized case, where this latent factor is the only driver of the whole phenomenon, the $\elij{E}{cp}$ values will be constant within each co-cluster, leading to two step-like estimators: $\eli{\ceig}{2}$ and $\eli{\peig}{2}$.

However, a latent factor interpretation of the \ecipci{} can still be valid even though $\mtx{E}$ is not block-structured. Let us define a bi-incidence matrix $[\mtx{R},\mtx{C}]^{\transp}$ where
\begin{equation}
  \mtx{R} = \begin{blockarray}{ccccc}
    \phantom{} & c_1 & c_2 & \cdots & c_m \\
    \begin{block}{c[cccc]}
      e_1 & 1 & 0 & \cdots & 0 \\
      e_2 & 0 & 1 & \cdots & 0 \\
      \vdots & \cdots & \vdots & \ddots & \vdots \\
      e_k & 0 & 0 & \cdots & 1 \\
    \end{block}
  \end{blockarray}
   ,\qquad
  \mtx{C} = \begin{blockarray}{ccccc}
    \phantom{} & p_1 & p_2 & \cdots & p_n \\
    \begin{block}{c[cccc]}
      e_1 & 1 & 0 & \cdots & 0 \\
      e_2 & 0 & 1 & \cdots & 0 \\
      \vdots & \cdots & \vdots & \ddots & \vdots \\
      e_k & 0 & 0 & \cdots & 0 \\
    \end{block}
  \end{blockarray}
  ,
\end{equation}
and both $R_{ki}$ and $C_{kj}$ assume value one if country $c_i$ competitively produces product $p_j$, and zero otherwise. The specialization bi-adjacency matrix is then defined as $\spec = \mtx{R}^{\transp} \, \mtx{C}$. We can also express $\mtx{R}$ and $\mtx{C}$ as
\begin{equation}
  \begin{cases}
    \mtx{R} \, \csv = \vect{\rho}, \\
    \mtx{C} \, \psv = \vect{\psi},
  \end{cases}
\end{equation}
where $\vect{\rho}$ and $\vect{\psi}$ can be called ``specialization'' and ``concentration scores'', while $\csv$ and $\psv$ are ``latent scores'' (please, remember that $\nth{\ceig}{2} = \nth{\cclust}{2} = \divty^{-1/2} \, \nth{\csv}{2}$ and $\nth{\peig}{2} = \nth{\svals}{2}^{-1} \, \nth{\pclust}{2} = \nth{\svals}{2}^{-1} \, \ubity^{-1/2} \, \nth{\psv}{2}$). Therefore, the \ecipci{} is a special case of Canonical Correlation (CCA).\cite{HotellingRelations1936, UurtioEtAlTutorial2017, AbdiEtAlCanonical2018} Accordingly, the phenomenon represented by $\grp{\specs}$ can be organized in two complementary views. From the viewpoint of countries, they are specialization patterns represented by $\mtx{R}$. While, from products' viewpoint, they are concentration patterns represented by $\mtx{C}$. From a clustering perspective, the key idea of CCA is that the information embedded in the different views on the same phenomenon is complementary, and analyzing the relationship between these views conveys the underlying clusters. Specifically, the ECI and PCI are optimal mappings between the two views since $\csv$ and $\psv$ maximize the covariance between $\vect{\rho}$ and $\vect{\psi}$ under the constraints that these last are centered and have unit variance.\cite{GreenacreTheory1984, FoussEtAlBipartite2016}

Therefore, if we observe that country $c$ is specialized in product $p$ (or that $p$ is concentrated in $c$, as seen before) and believe in a latent factor driving both the specialization and concentration phenomena, we should expect $c$ and $p$ to be associated in the ECI-PCI latent space. We can also expect $c$ and $c'$ to be as close in the ECI-PCI latent space the more they have specialized in similar products, once compensated for disparities between their \emph{diversity} and for some random noise in the data. Figure~\ref{fig:country_product_probs_cca}e shows the latent space spanned by $\eli{\ceig}{2}$ and $\eli{\peig}{2}$ and highlights how, in this space, countries and products are substantially and positively associated, as the cloud of points lies along the diagonal line. It also shows, for each product $p$, an orange triangle corresponding to the average ECI of countries where it is competitively produced, while, for each country $c$, the purple squares represent the average PCI of the products it competitively produces. The purple squares lie on the diagonal of the space spanned by ECI and PCI since it is possible to show that the ECI of a country is equivalent to the average PCI of the products it competitively produces (i.e., $\nth{\ceigs}{2} = \divty^{-1} \, \spec \, \nth{\peigs}{2}$); see SI. However, the same is not true for the average ECI values (orange triangles) since the latent factor driving both specialization and concentration processes does not explain the whole variance in the data. Indeed, since $\nth{\peigs}{2} = \nth{\svals}{2}^{-2} \, \ubity^{-1} \, \spec^{\transp} \, \nth{\ceigs}{2}$, only if $\nth{\svals}{2} = 1$ the PCI of each product would be equal to the average ECI of the countries that competitively produce it, too. This happens when the whole variance in the data is captured by the space spanned by ECI and PCI. Therefore, Figure~\ref{fig:country_product_probs_cca}e also helps illustrate that the \ecipci{} maximizes the correlation strength (proportional to $\nth{\svals}{2}^{-2}$) between the two latent factor estimators, one for the countries and one for the products.

Therefore, a latent factor interpretation of the \ecipci{} reconciles the more recent similarity-based interpretation of the indices with the original capabilities-based one. However, it is beyond the scope of this paper to establish whether the ECI and PCI are proper measurements of organizational and technological capabilities or whether they may be interpreted as expressions of some other latent factor that explains both the specialization and concentration processes.

\section*{Data availability}
Data and code to reproduce the analysis of the article are available upon request.


\section*{Author information}

\subsection*{Contributions}
All authors contributed equally to the paper. C.B., J.Di~I., and M.I.: Conceptualization; Data curation; Formal analysis; Investigation; Methodology; Software; Visualization; Writing – original draft; Writing – review \& editing.

\section*{Ethics declarations}

\subsection*{Competing interests}
The authors declare no competing interests.

\clearpage
\appendix
\renewcommand{\thesection}{S\arabic{section}}
\renewcommand{\thefigure}{S\arabic{figure}}
\renewcommand{\thetable}{S\arabic{table}}
\renewcommand{\theequation}{S\arabic{equation}}

\vskip-36pt%
{\raggedright\sffamily\fontsize{12}{16}\selectfont Supplementary Information for\par}%
\vskip10pt%
{\raggedright\sffamily\bfseries\fontsize{20}{25}\selectfont Reinterpreting Economic Complexity: A co-clustering approach\par}%
\vskip10pt%
{\raggedright\fontsize{12}{12}\usefont{OT1}{phv}{b}{n} Carlo Bottai\textsuperscript{1}, Jacopo Di Iorio\textsuperscript{2} \& Martina Iori\textsuperscript{3}\par}%
\vskip10pt%
{\raggedright\fontsize{10}{12}\usefont{OT1}{phv}{m}{n} 1~Department of Economics, Management and Statistics, University of Milano–Bicocca, Milano, Italy\par}%
{\raggedright\fontsize{10}{12}\usefont{OT1}{phv}{m}{n} 2~Department of Statistics, Penn State University, University Park (PA), USA\par}%
{\raggedright\fontsize{10}{12}\usefont{OT1}{phv}{m}{n} 3~Institute of Economics \& L'EMbeDS, Sant'Anna School of Advanced Studies, Pisa, Italy\par}%

\section*{Supplementary Methods}

\section{Equivalence between the clustering and co-clustering approaches}
In the following two propositions, we prove the connection between the $\ecis$ (and the $\pcis$) and the singular vector decomposition of the matrix $\sym{\specs} = \divty^{-1/2} \, \spec \, \ubity^{-1/2}$. These two results are pivotal to rewrite the \ecipci{} algorithm as a co-clustering algorithm that is indeed based on the singular vectors of $\sym{\specs}$ rescaled according to $\divty^{-1/2}$ and $\ubity^{-1/2}$.

\begin{proposition} \label{prop:eci}
  Given a matrix $\spec$, its (non-standardized) $\ecis$ is $\nth{\cclusts}{2} = \divty^{-1/2} \, \nth{\csv}{2}$ where $\nth{\csv}{2}$ is the second left singular vector of the matrix $\sym{\specs} = \divty^{-1/2} \, \spec \, \ubity^{-1/2}$.
\end{proposition}

\begin{proof}[Proof]
It is well known that for any matrix $\mtx{A}$, their left singular vectors correspond to the eigenvectors of $\mtx{A} \, \mtx{A}^{\transp}$, with eigenvalues equal to the square of the corresponding singular value.
Therefore, the left singular vectors $\csv$ of $\sym{\specs} = \divty^{-1/2} \, \spec \, \ubity^{-1/2}$ correspond to the eigenvectors of $\sym{\specs} \, \sym{\specs}^{\transp}$.

\begin{align*}
    \sym{\specs} \sym{\specs}^{\transp}
    &= \divty^{-1/2} \, \spec \, \ubity^{-1/2} (\divty^{-1/2} \, \spec \, \ubity^{-1/2})^{\transp} \\
    &= \divty^{-1/2} \, \spec \, \ubity^{-1/2}  (\ubity^{-1/2})^{\transp} \, \spec^{\transp} \, (\divty^{-1/2})^{\transp} \\
    &= \divty^{-1/2} \, \spec \, \ubity^{-1/2}  \ubity^{-1/2} \, \spec^{\transp} \, \divty^{-1/2} \\
    &= \divty^{-1/2} \, \spec \, \ubity^{-1} \, \spec^{\transp} \, \divty^{-1/2} \\
    &= \divty^{-1/2} \csim \divty^{-1/2}
\end{align*}
where $\csim = \spec \, \ubity^{-1} \, \spec^{\transp}$. Therefore, the left singular vectors $\csv$ of the matrix $\sym{\spec}$ are equal to the eigenvectors of $\divty^{-1/2} \, \csim \, \divty^{-1/2}$.

In other words, we can write
\begin{equation*}
(\divty^{-1/2} \, \csim \, \divty^{-1/2}) \, \csv = (1-\eigvals) \, \csv
\end{equation*}

Multiplying by $\divty^{-1/2}$ on both the sides of the equation:
\begin{align*}
  & \divty^{-1/2} \, (\divty^{-1/2} \, \csim \, \divty^{-1/2}) \, \csv = (1-\eigvals) \, \divty^{-1/2} \, \csv \\
  & \divty^{-1} \, \csim \, (\divty^{-1/2} \, \csv) = (1-\eigvals) \, (\divty^{-1/2} \, \csv) \\
  & \cproj \, (\divty^{-1/2} \, \csv) = (1-\eigvals) \, (\divty^{-1/2} \, \csv) \\
  & \cproj \, \cclust = (1-\eigvals) \, \cclust
\end{align*}

The eigenvectors $\cclust$ of $\cproj$ are equal to $\divty^{-1/2} \, \csv$, where $\csv$ are the left singular vectors of the matrix $\sym{\specs}$. Therefore, being the (non-standardized) ECI defined as the Fiedler eigenvector of $\cproj$, $\nth{\ceigs}{2}  = \divty^{-1/2} \, \nth{\csv}{2}$.
\end{proof}

\begin{proposition} \label{prop:pci}
  Given a matrix $\spec$, its (non-standardized) $\pcis$ is $\nth{\peigs}{2} = \nth{\svals}{2}^{-1} \, \ubity^{-1/2} \, \nth{\psv}{2}$ where $\nth{\psv}{2}$ is the second right singular vector of the matrix $\sym{\specs} = \divty^{-1/2} \, \spec \, \ubity^{-1/2}$ and $\nth{\svals}{2}$ the corresponding singular value.
\end{proposition}

\begin{proof}[Proof]
It is well known that for any matrix $\mtx{A}$, their right singular vectors correspond to the eigenvectors of $\mtx{A}^{\transp} \, \mtx{A}$.
Therefore, the right singular vectors $\nth{\psv}{i}$ of $\sym{\specs} = \divty^{-1/2} \, \spec \, \ubity^{-1/2}$ correspond to the eigenvectors of $\sym{\specs}^{\transp} \, \sym{\specs}$.

\begin{align*}
  \sym{\specs}^{\transp} \, \sym{\specs}
  &= (\divty^{-1/2} \, \spec \, \ubity^{-1/2})^{\transp} \, (\divty^{-1/2} \, \spec \, \ubity^{-1/2}) \\
  &= (\ubity^{-1/2})^{\transp} \, \spec^{\transp} \, (\divty^{-1/2})^{\transp} \, (\divty^{-1/2} \, \spec \, \ubity^{-1/2}) \\
  &= \ubity^{-1/2} \, \spec^{\transp} \, \divty^{-1/2} \, \divty^{-1/2} \, \spec \, \ubity^{-1/2} \\
  &= \ubity^{-1/2} \, \spec^{\transp} \, \divty^{-1} \, \spec \, \ubity^{-1/2} \\
  &= \ubity^{-1/2} \, \psim \, \ubity^{-1/2}
\end{align*}

where $\psim = \spec^{\transp} \, \divty^{-1} \, \spec$.
Therefore, the right singular vector $\psv$ of the matrix $\sym{\specs}$ are equal to the eigenvectors of $\ubity^{-1/2} \, \psim \, \ubity^{-1/2}$.

In other words, we can write
\begin{equation*}
  (\ubity^{-1/2} \, \psim \, \ubity^{-1/2}) \psv = (1-\eigvals) \, \psv
\end{equation*}

Multiplying by $\ubity^{-1/2}$ on both sides of the equation:
\begin{align*}
  & \ubity^{-1/2} \, (\ubity^{-1/2} \, \psim \, \ubity^{-1/2}) \, \psv = (1-\eigvals) \, \ubity^{-1/2} \, \psv \\
  & \ubity^{-1} \, \psim \, (\ubity ^{-1/2} \, \psv) = (1-\eigvals) \, (\ubity^{-1/2} \, \psv) \\
  & \pproj \, (\ubity^{-1/2} \, \psv) = (1-\eigvals) \, (\ubity^{-1/2} \, \psv) \\
  & \pproj \, \pclust = (1-\eigvals) \, \pclust
\end{align*}

The eigenvectors $\pclust$ of $\pproj$ are equal to $\ubity^{-1/2} \, \psv$, where $\psv$ are the right singular vectors of the matrix $\sym{\specs}$. Therefore, being the (non-standardized) PCI defined as the Fiedler eigenvector of $\pproj$ scaled by $\nth{\svals}{2}^{-1}$, $\nth{\peigs}{2} = \nth{\svals}{2}^{-1} \, \ubity^{-1/2} \, \nth{\psv}{2}$
\end{proof}

In conclusion, the \ecipci{} results from the singular vectors of $\sym{\spec}$ rescaled according to $\divty^{-1/2}$ and $\nth{\svals}{2}^{-1} \, \ubity^{-1/2}$, respectively.

\section{Transition equations}

\subsection{Interrelation between non-standardized ECI and PCI}
As mentioned, the MoR implies that the ECI of a country corresponds to the average PCI of the products that this country competitively produces.\cite{HidalgoHausmannBuilding2009} Indeed, as already shown in the literature,\cite{HillCorrespondence1974, MealyEtAlInterpreting2019} we can state the following proposition that can be proven both in clustering and co-clustering frameworks.

\begin{proposition} \label{prop:eigens_relation}
  The eigenvectors of $\cproj$ are equal to the \textcolor{blue}{weighted} average of the eigenvectors of $\pproj$. That is, $\nth{\ceigs}{i} \equiv \divty^{-1} \, \spec \, \nth{\peigs}{i}$.
\end{proposition}

\begin{proof}[Proof in a clustering context]
  Remember that $\cproj = \divty^{-1} \, \spec \, \ubity^{-1} \, \spec^{\transp}$ and $\pproj = \ubity^{-1} \, \spec^{\transp} \, \divty^{-1} \, \spec$. Moreover, remember that we defined $\nth{\ceigs}{i}$ and $\nth{\peigs}{i}$ as solutions of the following eigenvalue problems: $\cproj \, \ceig = (1 - \eigvals) \, \ceig$ and $\pproj \, \peig = (1 - \eigvals) \, \peig$.
  \begin{align*}
    \cproj \, \nth{\ceigs}{i} &= (1 - \nth{\eigvals}{i}) \, \nth{\ceigs}{i}, \\
    (\spec^{-1} \, \divty) \, \cproj \, \nth{\ceigs}{i} &= (1 - \nth{\eigvals}{i}) \, (\spec^{-1} \, \divty) \, \nth{\ceigs}{i}, \\
    (\spec^{-1} \, \divty) \, (\divty^{-1} \, \spec \, \ubity^{-1} \, \spec^{\transp}) \, \nth{\ceigs}{i} &= (1 - \nth{\eigvals}{i}) \, (\spec^{-1} \, \divty) \, \nth{\ceigs}{i}, \\
    (\ubity^{-1} \, \spec^{\transp}) \, \nth{\ceigs}{i} &= (1 - \nth{\eigvals}{i}) \, (\spec^{-1} \, \divty) \, \nth{\ceigs}{i}, \\
    (\ubity^{-1} \, \spec^{\transp}) \, (\divty^{-1} \, \spec \, \nth{\peigs}{i}) \, &= (1 - \nth{\eigvals}{i}) \, (\spec^{-1} \, \divty) \, (\divty^{-1} \, \spec \, \nth{\peigs}{i}), \\
    \pproj \, \nth{\peigs}{i} &= (1 - \nth{\eigvals}{i}) \, \nth{\peigs}{i}.
  \end{align*}
  That is, both the stated problems are solved by the same eigenvalues if $\nth{\ceigs}{i} \equiv \divty^{-1} \, \spec \nth{\peigs}{i}$.
\end{proof}

\begin{proof}[Proof in a co-clustering context]
  The SVD of $\sym{\specs}$ consists of solving the following system of equations:
  \begin{equation*}
    \begin{cases}
      \big( \divty^{-1/2} \, \spec \, \ubity^{-1/2} \big) \, \psv &= \svals \, \csv, \\
      \big( \ubity^{-1/2} \, \spec^{\transp} \, \divty^{-1/2} \big) \, \csv &= \svals \, \psv.
    \end{cases}
  \end{equation*}
  Since we proved that $\nth{\ceigs}{i} = \divty^{-1/2} \, \nth{\csv}{i}$ and $\nth{\peigs}{i} = \nth{\svals}{i}^{-1} \, \ubity^{-1/2} \, \nth{\psv}{i}$, we can substitute them into this system of equations.
  \begin{equation*}
    \begin{cases}
      \big( \divty^{-1/2} \, \spec \, \ubity^{-1/2} \big) \, \big( \nth{\svals}{i} \, \ubity^{1/2} \, \nth{\peig}{i} \big) &= \nth{\svals}{i} \, \big( \divty^{1/2} \, \nth{\ceig}{i} \big), \\
      \big( \ubity^{-1/2} \, \spec^{\transp} \, \divty^{-1/2} \big) \, \big( \divty^{1/2} \, \nth{\ceig}{i} \big) &= \nth{\svals}{i} \, \big( \nth{\svals}{i} \, \ubity^{1/2} \, \nth{\peig}{i} \big).
    \end{cases}
  \end{equation*}
  \begin{equation} \label{eq:ecipci_relationship}
    \begin{cases}
      \divty^{-1} \, \spec \, \nth{\peig}{i} &= \nth{\ceig}{i}, \\
      \ubity^{-1} \, \spec^{\transp} \nth{\ceig}{i} &= \nth{\svals}{i}^{2} \, \nth{\peig}{i}.
    \end{cases}
  \end{equation}
\end{proof}

This proof is equivalent to the previous one, but the simultaneity of the two solutions is now more evident. Further, it shows that both equivalences in equation (\ref{eq:ecipci_relationship}) must simultaneously hold if we want to solve $\cproj \, \ceig = (1 - \eigvals) \, \ceig$ and $\pproj \, \peig = (1 - \eigvals) \, \peig$ at the same time, remembering, from the correspondence between the SVD of $\sym{\specs}$ and the eigen-decomposition of $\sym{\specs} \, \sym{\specs}^{\transp}$, that $\svals^2 = 1 - \lambda$.

\subsection{Interrelation between ECI and PCI}
However, the ECI is usually defined as the \emph{standardized} second eigenvalue of $\cproj$.\cite{HausmannEtAlAtlas2014} Therefore, to preserve the equivalence between ECI and ``average PCI'',\cite{HidalgoHausmannBuilding2009} the PCI is usually defined as the second eigenvalue of $\pproj$ standardized using the same mean and variance used to define the ECI.\cite{CIDEcomplexity2014} It follows another proposition.
\begin{proposition} \label{prop:ecipci_relation}
  We define $\ecis \coloneq \big( \nth{\ceigs}{2} - \mean{\nth{\ceigs}{2}} \big) \big/ \sd{\nth{\ceigs}{2}}$ and $\pcis \coloneq \big( \nth{\peigs}{2} - \mean{\nth{\ceigs}{2}} \big) \big/ \sd{\nth{\ceigs}{2}}$, where $\mean{\vect{\cdot}}$ and $\sd{\vect{\cdot}}$ are the mean and standard deviation of the corresponding vector, so that the ECI of a country is equal to the average PCI of the products it produces competitively. That is, $\ecis \equiv \divty^{-1} \, \spec \cdot \pcis$.
\end{proposition}
\begin{proof}
  We proved that $\nth{\ceigs}{2} \equiv \divty^{-1} \, \spec \, \nth{\peigs}{2}$. Therefore,
  \begin{align*}
    \ecis &\equiv \divty^{-1} \, \spec \cdot \pcis, \\
               &= \divty^{-1} \, \spec \, \big( \nth{\peigs}{2} - \mean{\nth{\ceigs}{2}} \big) \big/ \sd{\nth{\ceigs}{2}}, \\
               &= \big( \divty^{-1} \, \spec \, \nth{\peigs}{2} - \divty^{-1} \, \spec \, \mean{\nth{\ceigs}{2}} \big) \big/ \sd{\nth{\ceigs}{2}}, \\
               &= \big( \nth{\ceigs}{2} - \mean{\nth{\ceigs}{2}} \big) \big/ \sd{\nth{\ceigs}{2}}.
  \end{align*}
  Please, notice that $\divty^{-1} \, \spec \, \mean{\nth{\ceigs}{2}} = \mean{\nth{\ceigs}{2}}$ because $\frac{\sum_p \elij{\specs}{cp} \, \mean{\nth{\ceigs}{2}}}{\sum_p \elij{\specs}{cp}} = \mean{\nth{\ceigs}{2}} \;\; \forall \, c$.
\end{proof}

Yet, since $\pcis \not\equiv \nth{\svals}{i}^{-2} \, \ubity^{-1} \, \spec^{\transp} \, \ecis$, this breaks equation~(\ref{eq:ecipci_relationship}), restating that the ECI-PCI framework has been focusing on countries much more than on products since its very beginning.

\section*{Supplementary Tables}

\section{Main countries and products in low- and high-complexity co-clusters}

\begin{table}[hp]
  \centering
  \begin{tabular}{l|>{\centering\arraybackslash}p{1.5cm}>{\centering\arraybackslash}p{1.8cm}|>{\centering\arraybackslash}p{1.5cm}>{\centering\arraybackslash}p{1.8cm}|>{\centering\arraybackslash}p{1.5cm}>{\centering\arraybackslash}p{1.8cm}|>{\centering\arraybackslash}p{1.5cm}>{\centering\arraybackslash}p{1.8cm}}
  \toprule
    & \multicolumn{2}{c}{1970} & \multicolumn{2}{c}{1985} & \multicolumn{2}{c}{2000} & \multicolumn{2}{c}{2015} \\
  \midrule
    & Country & Product & Country & Product & Country & Product & Country & Product \\
  \midrule
    $\mathcal{A}$ & Sudan, Libya, Nigeria, Uganda, Indonesia & Sawlogs, Palm oil, Nat.~rubber, Palm nuts, Castor oil & Nigeria, Brunei, Cameroon, Congo, Gabon & Sawlogs, Cocoa, Coffee, \quad Nat.~rubber, Palm nuts & Iraq, Eq.~Guinea, Nigeria, G.-Bissau, Brunei & Min.~fuels, Petroleum, Cocoa, Nat.~rubber, Castor Oil & Iraq, Eq.~Guinea, G.-Bissau, Gabon, Nigeria & Cocoa, Petroleum, Nat.~rubber, Tin ores, Sesame \\
  \midrule
    $\mathcal{B}$ & Germany, Switz., Sweden, GB, Austria & Specialized machines, Elec.~motor, Dish-wash.~machines & Japan, Germany, Sweden, Switz., GB & Television, Opt.~instr., Music instr., Ind.~ovens, Tractors  & Japan, Germany, Switz., Sweden, GB & Railways, Trucks, Car ind., Metal ind., Spec.~machines & Japan, Taiwan, Switz.,\quad S.~Korea, Germany & Opt.~instr., Photogr., Metal ind., Chem.~ind., Spec.~mach.\\ 
  \bottomrule
  \end{tabular}
  \caption{Examples of countries and products in low- and high-complexity co-clusters $\mathcal{A}$ and $\mathcal{B}$ in different years. The table reports only countries and products with the highest probability of belonging to these co-clusters in the years 1970, 1985, 2000, and 2015. Products are classified according to Standard International Trade Classification (SITC) Rev. 2 at the 4-digit level.}
  \label{tab:top_bottom_clusters}
\end{table}

\end{document}